\documentclass[11pt,a4paper,reqno]{amsproc}
\usepackage{amsmath, amssymb, graphicx}
\usepackage{mathtools}
\DeclarePairedDelimiter{\norm}{\lVert}{\rVert}
\usepackage{csquotes}
\usepackage{adjustbox}
\usepackage{soul}
\usepackage{xcolor}
\usepackage{tikz}
\usepackage{float}
\usepackage{float}
\usepackage{booktabs}
\usepackage{tcolorbox}
\graphicspath{{Diagrams/}}
\usepackage[
top=1in,
bottom=1in,
left=1in,
right=1in
]{geometry}
\usepackage{listings}
\usepackage{caption}
\usepackage{hyperref}
\usepackage{float}
\usepackage{fancyhdr}

\usepackage{makecell}
\usetikzlibrary{shapes.geometric, arrows}
\tikzstyle{startstop} = [rectangle, rounded corners, minimum width=3cm, minimum height=1cm,text centered, draw=black, fill=white]
\tikzstyle{process} = [rectangle, minimum width=3cm, minimum height=1cm, text centered, draw=black, fill=white]
\tikzstyle{decision} = [diamond, minimum width=3cm, minimum height=1cm, text centered, draw=black, fill=white]
\tikzstyle{arrow} = [thick,->,>=stealth]
\usepackage{algorithm}
\usepackage{algpseudocode}
\newtheorem{theorem}{Theorem}

\title{Comparative Analysis of Compliance-Matrix Induced Norms in Structural Topology Optimization}
\author{Jyotiranjan Nayak$^1$}
\email{jyotiranjan\_n@srmap.edu.in}
\author{Shafeequdheen P$^1$}
\email{shafeequdheen\_p@srmap.edu.in}
\author{Vijayakrishna Rowthu$^1$}
\email{vijayakrishna.r@srmap.edu.in}

\address{$^1$Department of Mathematics, SRM University AP, Neerukonda, Mangalagiri Mandal, Mangalagiri, Andhra Pradesh 522240, India}

\begin{document}
	
	\maketitle
	
\begin{abstract}
	\noindent
	Compliance minimization is a central objective in structural topology optimization, commonly interpreted as the total strain energy of a system. In this work, we examine the influence of alternative compliance formulations based on different norm representations of structural energy. Specifically, we consider three formulations: the classical quadratic compliance, its square-root form corresponding to an $\ell_2$ norm, and a spectral $\ell_1$-norm based formulation derived from the stiffness weighted displacement field. Although these formulations arise from the same stiffness displacement relationship, they generate markedly different optimization landscapes and result in distinct structural topologies. Numerical results indicate that the classical formulation produces well-distributed load paths, whereas the $\ell_1$-based formulation promotes sparse and highly localized structural members. These findings underscore the critical role of objective function selection in topology optimization and offer insights into alternative formulations for achieving tailored structural performance.
\end{abstract}
	
\section{Introduction}

Topology optimization has become a fundamental computational tool for determining optimal material distributions in structures subject to prescribed loads and boundary conditions. Among the various formulations, compliance minimization remains the most widely used objective due to its clear physical interpretation as total strain energy and its favorable numerical properties. Classical compliance is defined as a quadratic form involving the global stiffness matrix and the displacement vector. This formulation results in smooth optimization landscapes and well defined sensitivity expressions, which facilitate efficient numerical implementation. Despite these advantages, alternative formulations of compliance can be derived by interpreting strain energy within different normed spaces. In particular, the square root of the quadratic compliance can be interpreted as an $\ell_2$ norm in the stiffness-weighted displacement space. Furthermore, by exploiting the spectral decomposition of the stiffness matrix, an $\ell_1$-norm-based formulation can be constructed. These alternative definitions alter the geometric characteristics of the objective function and, consequently, influence the distribution of sensitivities throughout the optimization process. From a computational perspective, the dissemination and adoption of topology optimization methods have been significantly influenced by compact and accessible implementations. A notable example is the 99-line MATLAB \cite{MATLAB:2022b} code developed by Ole Sigmund \cite{sigmund200199}, which provides a minimal yet complete framework for compliance based topology optimization. Owing to its clarity, compactness, and ease of modification, this code has become a standard reference in both teaching and research, enabling rapid prototyping and testing of new ideas. In topology optimization, the most commonly used algorithms include the optimality criteria method (OCM) \cite{bendsoe1995optimization}, the method of moving asymptotes (MMA) \cite{svanberg1987method}, and sequential linear programming (SLP) \cite{dunning2015introducing}, owing to their simplicity and ease of implementation. A variant of the OCM, known as the generalized optimality criteria method (GOCM), was introduced by the authors in \cite{kim2021generalized}. In this approach, the volume fraction constraint is not strictly enforced at every iteration of the optimization process; however, it offers the advantage of handling multiple constraints more effectively, albeit with reduced computational efficiency.   

Building on this foundation, several researchers have proposed extensions and modifications to the original code to incorporate improved filtering techniques \cite{lazarov2011filters, bourdin2001filters, wang2008hybrid}, alternative material interpolation schemes \cite{bendsoe1999material, stolpe2001alternative}, and enhanced numerical stability \cite{jog1996stability, waszczyszyn2013stability}. Notable developments include the extension to a 88-line code with improved efficiency and filtering strategies \cite{andreassen2011efficient}, as well as further adaptations addressing issues such as mesh-independency, projection methods, and multi-physics formulations. These contributions demonstrate the flexibility of the original framework and its suitability as a baseline for methodological advancements.

In the present work, this established computational framework is further extended to incorporate alternative compliance measures based on $\ell_2$ and $\ell_1$ norms. While retaining the core finite element formulation and filtering approach, the objective function and corresponding sensitivity analyses are modified to reflect these alternative norm-based interpretations. This enables a consistent and controlled comparison between the classical quadratic compliance and its variants, thereby isolating the effect of norm selection on the resulting optimized topologies. The primary objective of this study is to investigate the influence of these alternative compliance measures on structural designs. By systematically comparing the classical formulation with its $\ell_2$ and $\ell_1$ counterparts, this work aims to provide deeper insight into how the choice of norm impacts sensitivity distribution and, ultimately, topology optimization outcomes.

The remainder of this paper is organized as follows. Section 2 presents the necessary preliminaries, including the governing equations of linear elasticity in their strong and weak forms, along with the formulation of the topology optimization problem based on the Solid Isotropic Material with Penalization (SIMP) approach. Section 3 introduces the modified formulations adopted in this work, detailing the incorporation of alternative compliance measures based on $\ell_2$ and $\ell_1$ norms and their corresponding sensitivity analyses. Section 4 provides the numerical results and discussion, where the proposed approaches are evaluated and compared against the classical formulation. Finally, Section 5 summarizes the main findings and presents the conclusions of the study.
	
\section{Preliminaries}	
The strong form of the  linear elasticity problem seeks a displacement field $\mathbf{u} : \Omega \to \mathbb{R}^d$ satisfying the following boundary value problem (see \cite{ern2004theory}):
\begin{align}
	\left\{
	\begin{aligned}
		-\nabla \cdot \sigma (\mathbf{u}) &= \mathbf{f}, && \text{in } \Omega, \\
		\mathbf{u} &= \mathbf{0}, && \text{on } \Gamma_D, \\
		\sigma(\mathbf{u}) \cdot \hat{n} &= \mathbf{g}, && \text{on } \Gamma_N,
	\end{aligned}
	\right.
	\label{1.1}
\end{align}
where $\mathbf{f}$ is the body force per unit volume, $\mathbf{g}$ is the prescribed surface traction, and $\hat{n}$ is the outward unit normal to the boundary. Further the constitutive equations are given by

\begin{align}
	\left\{
	\begin{aligned}
		\sigma(\mathbf{u}) &= A \varepsilon(\mathbf{u}) = \lambda\, \text{tr}(\boldsymbol{\varepsilon}(\mathbf{u}))\, \mathbf{I} + 2\mu\, \boldsymbol{\varepsilon}(\mathbf{u}), \\
		\boldsymbol{\varepsilon}(\mathbf{u}) &= \frac{1}{2} \left( \nabla \mathbf{u} + \nabla \mathbf{u}^\top \right).
	\end{aligned}
	\right.
	\label{1.2}
\end{align}
The stress tensor $\boldsymbol{\sigma}$ and the linearized strain tensor $\boldsymbol{\varepsilon}$ are related through Hooke’s law for isotropic materials (see \cite{johnson2009numerical}),
where $\lambda$ and $\mu$ are the Lamé parameters, $\text{tr}(\cdot)$ denotes the trace operator, and $\mathbf{I}$ is the identity tensor. The elasticity tensor $A$ is used in more general anisotropic formulations (see \cite{ciarlet1997mathematical}).

\subsection{Weak formulation}
Assuming the solution $\mathbf{u}$ is made up of a linear combination of some test functions from
\( V = \{ \mathbf{v} \in [H^1(\Omega)]^d : \mathbf{v}|_{\Gamma_D} = \mathbf{0} \} \). The weak form of the linear elasticity problem can be compactly written as the variational equation
\begin{equation}
	\int_\Omega \sigma(\textbf{u}): \varepsilon(\textbf{v}) = \int_{\Omega} \mathbf{f} \cdot \textbf{v} dx + \int_{\Gamma_N} \mathbf{g} \cdot \textbf{v} ds.
	\label{eq:main}
\end{equation}
\begin{equation}
	a(\textbf{u}, \textbf{v}) = L(\textbf{v})
	\label{eq:auvl}
\end{equation}
where the bilinear form $a(\cdot,\cdot)$ and the linear form $L(\cdot)$ are defined as in equation~\eqref{eq:a} and ~\eqref{eq:l} respectively.
\begin{equation}
	\begin{aligned}
		a(\mathbf{u}, \mathbf{v}) 
		&= \int_\Omega \sigma(\mathbf{u}) : \varepsilon(\mathbf{v}) \, dx = \int_{\Omega} \lambda (\nabla \cdot \mathbf{u})(\nabla \cdot \mathbf{v}) \, dx \quad + \int_{\Omega} 2\mu \boldsymbol{\varepsilon}(\mathbf{u}) 
		: \boldsymbol{\varepsilon}(\mathbf{v}) \, dx
	\end{aligned}
	\label{eq:a}
\end{equation}

\begin{equation}
	L(\textbf{v}) = \int_{\Omega} \mathbf{f} \cdot \textbf{v} dx + \int_{\Gamma_N} \mathbf{g} \cdot \textbf{v} ds.
	\label{eq:l}
\end{equation}

\noindent
Setting \(\textbf{u} = \sum_{j=1}^{N} C_j \phi_j\) and \(\textbf{v} = \phi_i\), where $\{\phi_i\}_{i=1}^{N}$
are the nodal basis functions that span the finite-dimensional space \(V_h \subset V\) and equation~\eqref{eq:auvl} can be written as system of linear equations by using equation~\eqref{1.2} as follows
\begin{equation}
	\begin{aligned}
		\sum_{j=1}^{N} \Big[ 
		&\lambda \int_{\Omega} (\nabla \cdot \phi_i)(\nabla \cdot \phi_j) \, dx + 2\mu \int_{\Omega} \boldsymbol{\varepsilon}(\phi_i) 
		: \boldsymbol{\varepsilon}(\phi_j) \, dx 
		\Big] C_j = \int_{\Omega} \mathbf{f} \cdot \phi_i \, dx 
		+ \int_{\Gamma_N} \mathbf{g} \cdot \phi_i \, ds
	\end{aligned}
	\label{eq:weak-9}
\end{equation}
\noindent
Further, the equation~\eqref{eq:weak-9} can be represented in the matrix form as  \( \mathbf{K}\mathbf{U} = \mathbf{F} \), where, \( K_{ij} = a(\phi_i, \phi_j)\), \(F_i = L(\phi_i) \) and \(\mathbf{U}\) consists of the unknowns \( C_j \), \( 1 \leq i, j \leq n \). The well-posedness and proof details can be found in Emma Cinatl \cite{cinatl2018finite}.

\subsection{SIMP for Topology Optimization Problem}

A widely used method for topology optimization is the SIMP (Solid Isotropic Material with Penalization) approach, originally proposed by Bends{\o}e \cite{Bendsoee1989} and later formalized with Sigmund \cite{Bendsoee2001}. In the SIMP method, each finite element is assigned a design variable $x_e \in [0,1]$, representing its relative material density.

The topology optimization problem is formulated as the minimization of compliance (i.e., maximization of structural stiffness) subject to a volume constraint:
\begin{align}
	\left\{
	\begin{aligned}
		\min_{\mathbf{x}} \quad & C(\mathbf{x}) = \mathbf{U}^\top \mathbf{K}(\mathbf{x}) \mathbf{U} 
		= \sum_{e=1}^{N} x_e^p\, \mathbf{u}_e^\top \mathbf{K}_e \mathbf{u}_e \\
		\text{subject to} \quad 
		& \dfrac{V(\mathbf{x})}{V_0} = f, \\
		& \mathbf{K}(\mathbf{x}) \mathbf{U} = \mathbf{F}, \\
		& x_{\min} \leq x_e \leq 1, \quad e = 1, \ldots, N.
	\end{aligned}
	\right.
	\label{eq:simp}
\end{align}

\noindent
Here, $\mathbf{x} = \{x_1, x_2, \ldots, x_N\}^\top$ is the vector of design variables. The total material volume is given by $V(\mathbf{x}) = \sum_{e=1}^{N} x_e v_e$, where $v_e$ denotes the volume of element $e$, and $V_0$ is the volume of the full design domain. The volume fraction $f \in (0,1)$ specifies the allowable material usage. The lower bound $x_{\min}$ is introduced to ensure numerical stability and to avoid singularities in the global stiffness matrix $\mathbf{K} \in \mathbb{R}^{n \times n}$, where $\mathbf{U}$ and $\mathbf{F}$ denote the global displacement and force vectors, respectively.

\section{Modified SIMP Model}

In addition to the classical quadratic compliance, alternative formulations can be constructed by introducing different norm definitions in appropriately transformed displacement spaces. In the standard SIMP framework, the element compliance is expressed as a quadratic form of the displacement field, which can be equivalently interpreted as the squared $\ell_2$ norm of a stiffness weighted displacement vector. This observation naturally motivates a broader class of compliance measures obtained by replacing the quadratic form with general vector norms in the transformed space. While these formulations retain the physical interpretation of strain energy, they modify the geometric structure of the objective function particularly its scaling, convexity, and sensitivity distribution which in turn influences the optimization trajectory and resulting material layout. In this section, two such compliance measures based on $\ell_2$ and $\ell_1$ norms are presented.
\subsection{$\ell_2$-Norm Compliance}

An alternative compliance measure can be obtained by taking the square root of the classical quadratic compliance and the optimization modifies into the following. 
 
\begin{align}
	\left\{
	\begin{aligned}
		\min_{\mathbf{x}} \quad & 	C_{2}(x) = \sum_{e=1}^{N} x_e^p\,\sqrt{\mathbf{u}_e^\top \mathbf{K}_e \mathbf{u}_e}.\\
		\text{subject to} \quad & \dfrac{V(\mathbf{x})}{V_0} = f,\\ & \mathbf{K}(\mathbf{x}) \mathbf{U} = \mathbf{F}, \\ & 0 \leq x_{\min} \leq x_e \leq 1, \quad e = 1, \ldots, N.
	\end{aligned}
	\right.
	\label{eq:simpl2}
\end{align}
\noindent
This expression admits a natural interpretation as an $\ell_2$ norm in a stiffness-weighted displacement space:
\begin{equation}
	C_{2} = \norm{\mathbf{K}^{\frac{1}{2}} \mathbf{U}}_2.
\end{equation}
\noindent
Although $C_{2}$ is a monotonic transformation of the classical compliance, it introduces a nonlinear scaling that alters the relative magnitude of sensitivities. As a result, while the optimal solution in an ideal continuous setting remains unchanged, the optimization trajectory and convergence characteristics may differ in practical implementations.

\subsection{$\ell_1$-Norm Compliance}

A fundamentally different formulation can be derived by considering a spectral decomposition of the element stiffness matrix. For each element, the stiffness matrix can be expressed as:
\begin{equation}
	\mathbf{K}_e = \mathbf{V} \mathbf{D} \mathbf{V}^\top,
\end{equation}
where $\mathbf{V}$ is the matrix of eigenvectors and $\mathbf{D}$ is a diagonal matrix containing the corresponding eigenvalues. The square root of the element stiffness matrix is then given by:
\begin{equation}
	\mathbf{K}_e^{\frac{1}{2}} = \mathbf{V} \mathbf{D}^{\frac{1}{2}} \mathbf{V}^\top.
\end{equation}

Defining a transformed displacement vector at the element level as 
\begin{equation}
	\mathbf{w}_e = \mathbf{K}_e^{\frac{1}{2}} \mathbf{U}_e,
\end{equation}
the corresponding compliance measure can be written as ~\eqref{eq:siml2} with the same constraints as described in ~\eqref{eq:simp} and ~\eqref{eq:simpl2}.
\begin{equation}
	C_{1} = \sum_{e} \|\mathbf{w}_e\|_1 = \sum_{e} \sum_{i} |w_{e,i}|.
	\label{eq:siml2}
\end{equation}
\noindent
In the present implementation, the matrix square root is computed via eigenvalue decomposition of the element stiffness matrix. It has been noticed that eigen values of the local stiffness matrix exhibits very small and negative values in the classical SIMP. To maintain the local stiffness matrices as positive definite, pre-conditioning \cite{saad2011numerical} has been done where the eigenvalues smaller than a prescribed threshold $\varepsilon$ are truncated to zero prior to taking the square root. The transformed displacement vector is then evaluated using the resulting square root matrix, and the element-wise $\ell_1$-norm is obtained by summing the absolute values of its components. The procedure is summarized below:

\begin{algorithm}[H]
	\caption{Element-wise computation of $\ell_1$-norm compliance}
	\begin{algorithmic}[1]
		
		\For{each element $e$}
		\State $[\mathbf{V}, \mathbf{D}] = \mathrm{eig}(\mathbf{K}_e)$
		\State $D_{ii} \gets 0 \;\; \text{if } |D_{ii}| < \varepsilon$
		
		\State $\mathbf{K}_e^{\frac{1}{2}} = \mathbf{V} \mathbf{D}^{\frac{1}{2}} \mathbf{V}^\top$
		\State $\mathbf{w}_e = \mathbf{K}_e^{\frac{1}{2}} \mathbf{U}_e$
		\State $C_e = \|\mathbf{w}_e\|_1$
		\EndFor
		
		\State $C_{1} = \sum_e C_e$
		
	\end{algorithmic}
\end{algorithm}

The following theorem establishes the equivalence of the $\ell_1$ and $\ell_2$ norms in finite-dimensional spaces, which forms the basis for relating different compliance measures.
\begin{theorem}[Equivalence of $\ell_1$ and $\ell_2$ norms in $\mathbb{R}^n$]
	Let $\mathbf{w}_e \in \mathbb{R}^n$. Then the $\ell_1$ and $\ell_2$ norms are equivalent; in particular, the following inequalities hold:
	\begin{equation}
		\|\mathbf{w}_e\|_2 \;\le\; \|\mathbf{w}_e\|_1 \;\le\; \sqrt{n}\,\|\mathbf{w}_e\|_2.
	\end{equation}
	Consequently, there exist positive constants $c_1 = 1$ and $c_2 = \sqrt{n}$ such that
	\begin{equation}
		c_1 \|\mathbf{w}_e\|_2 \;\le\; \|\mathbf{w}_e\|_1 \;\le\; c_2 \|\mathbf{w}_e\|_2, \quad \forall \mathbf{w}_e \in \mathbb{R}^n.
	\end{equation}
\end{theorem}

\begin{proof}
	The proof follows from standard inequalities between vector norms in finite-dimensional spaces (see, e.g., Kreyszig~\cite{kreyszig1991introductory}).
	
	First, by the Cauchy--Schwarz inequality,
	\begin{equation}
		\|\mathbf{w}_e\|_1 = \sum_{i=1}^n |w_{e,i}| 
		\le \sqrt{n} \left( \sum_{i=1}^n w_{e,i}^2 \right)^{1/2} 
		= \sqrt{n}\,\|\mathbf{w}_e\|_2.
	\end{equation}
	
	Next, we use the inequality
	\begin{equation}
		\sum_{i=1}^n w_{e,i}^2 \le \left( \sum_{i=1}^n |w_{e,i}| \right)^2,
	\end{equation}
	which implies
	\begin{equation}
		\|\mathbf{w}_e\|_2 \le \|\mathbf{w}_e\|_1.
	\end{equation}
	
	Combining the two inequalities yields
	\begin{equation}
		\|\mathbf{w}_e\|_2 \le \|\mathbf{w}_e\|_1 \le \sqrt{n}\,\|\mathbf{w}_e\|_2,
	\end{equation}
	which completes the proof.
\end{proof}

\subsection{Modal Sensitivity and Mechanical Relevance}
	To further interpret the effect of the proposed norm-based compliance measures, consider the spectral decomposition of the element stiffness matrix:
	\begin{equation}
		\mathbf{K}_e = \mathbf{V} \mathbf{D} \mathbf{V}^\top,
	\end{equation}
	where $\mathbf{D} = \mathrm{diag}(\lambda_i)$ contains the eigenvalues and $\mathbf{V}$ the corresponding eigenvectors. Defining the modal displacement vector $\tilde{\mathbf{u}}_e = \mathbf{V}^\top \mathbf{U}_e$, the transformed variable becomes
	\begin{equation}
		\mathbf{w}_e = \mathbf{K}_e^{\frac{1}{2}} \mathbf{U}_e = \mathbf{V} \mathbf{D}^{\frac{1}{2}} \tilde{\mathbf{u}}_e,
	\end{equation}
	so that each component satisfies
	\begin{equation}
		w_{e,i} = \sqrt{\lambda_i}\,\tilde{u}_{e,i}.
	\end{equation}
	
	In this representation, the classical compliance per element can be written as
	\begin{equation}
		C(e) = x_e^p \sum_i \lambda_{e,i} \tilde{u}_{e,i}^2,
	\end{equation}
	which corresponds to a quadratic aggregation of modal strain energies. The associated sensitivity with respect to the displacement field is
	\begin{equation}
		\frac{\partial C(e)}{\partial \mathbf{U}_e} = 2 x_e^p \mathbf{K}_e \mathbf{U}_e = 2 x_e^p \mathbf{V} \mathbf{D} \tilde{\mathbf{u}}_e,
	\end{equation}
	indicating that each modal contribution is weighted proportionally to its eigenvalue $\lambda_i$. Consequently, stiffer modes dominate the response, leading to a distributed utilization of deformation modes.
	
	For the $\ell_2$-norm compliance,
	\begin{equation}
		C_{2}(e) = x_e^p \left( \sum_i \lambda_i \tilde{u}_{e,i}^2 \right)^{\frac{1}{2}},
	\end{equation}
	the sensitivity becomes
	\begin{equation}
		\frac{\partial C_{2}(e)(e)}{\partial \mathbf{U}_e}
		= x_e^p \frac{\mathbf{K}_e \mathbf{U}_e}{\|\mathbf{w}_e\|_2}
		= x_e^p \mathbf{V} \frac{\mathbf{D} \tilde{\mathbf{u}}_e}{\|\mathbf{w}_e\|_2}.
	\end{equation}
	Here, the normalization by $\|\mathbf{w}_e\|_2$ reduces the relative dominance of large modal contributions, resulting in a more balanced participation of deformation modes.
	
	In contrast, the $\ell_1$-norm compliance is given by
	\begin{equation}
		C_{1}(e) = x_e^p \sum_i \sqrt{\lambda_i}\,|\tilde{u}_{e,i}|,
	\end{equation}
	with sensitivity
	\begin{equation}
		\frac{\partial C_{1}(e)}{\partial \mathbf{U}_e}
		= x_e^p \mathbf{K}_e^{1/2} \,\mathrm{sign}(\mathbf{w}_e)
		= x_e^p \mathbf{V} \mathbf{D}^{\frac{1}{2}} \,\mathrm{sign}(\mathbf{w}_e).
	\end{equation}
	In this case, the contribution of each mode scales with $\sqrt{\lambda_i}$ and depends only on the sign of the modal amplitude, rather than its magnitude. This significantly weakens the influence of large eigenvalues and promotes sparsity in the modal response. The above relations reveal that the choice of norm directly controls how strain energy is distributed across the eigenmodes of the structure. The classical quadratic compliance enforces a fully distributed energy participation, where all modes contribute proportionally to their stiffness. The $\ell_2$-norm formulation moderates this effect by reducing the dominance of high-energy modes, leading to a more uniform stress and deformation field. In contrast, the $\ell_1$-norm formulation promotes a sparse modal representation, in which only a few dominant deformation modes carry significant energy. From a structural perspective, this transition can be interpreted as a shift from continuum-like behavior to more localized load transfer mechanisms. While the classical formulation favors diffuse material layouts, the $\ell_1$-based compliance encourages the formation of distinct load paths, often resulting in truss-like topologies. Thus, although all formulations are equivalent in a norm sense, they induce fundamentally different mechanical responses through their influence on modal energy distribution.

The overall optimization procedure follows the classical 99-line MATLAB\cite{MATLAB:2022b} code developed by Ole Sigmund\cite{sigmund200199}, with the only modification being in the compliance calculation and its corresponding sensitivity. For clarity, a schematic representation of the methodology adopted in this model is provided in Figure~\ref{fig:pseudocode_flowchart}, along with the associated pseudo-code describing the optimization loop.

\begin{figure}[htbp]
	\centering
	\begin{minipage}[htbp]{0.48\textwidth} 
		\begin{tcolorbox}[colback=gray!5,colframe=black!50,
			title=\textbf{Pseudo-code for the loop}, width=\textwidth,
			left=2mm, right=2mm, top=1mm, bottom=1mm,
			boxsep=1mm, halign=left]
			\hspace*{1em}~Initialize $x$ with volfrac for each element;\\ \hspace*{1em}~ $\texttt{loop}=0$, $\texttt{rchange}=1$\\
			~\textbf{while} $\texttt{rchange}>0.01$ \textbf{do}\\
			\hspace*{1em}~$\texttt{loop}\leftarrow\texttt{loop}+1$;  $x_{\text{old}}\leftarrow x$\\
			\hspace*{1em}~Compute displacement $\mathbf{U}\leftarrow\text{FEAnalysis}$\\
			\hspace*{1em}~$c\leftarrow0$\\
			\hspace*{1em}~\textbf{for} each element $e$ \textbf{do}\\
			\hspace*{2em}Compute $\mathbf{K}_e$, $\mathbf{U}_e$, 
			$\text{comp}\!\leftarrow\!\mathbf{U}_e^T\mathbf{K}_e\mathbf{U}_e$\\
			\hspace*{2em}$c\!\leftarrow\!c+(x_e)^{\texttt{penal}}\text{comp}$; \\ 
			\hspace*{2em}$dc_e\!\leftarrow\!-\texttt{penal}(x_e)^{\texttt{penal}-1}\text{comp}$\\
			\hspace*{1em}~\textbf{end for}\\
			\hspace*{1em}~$dc\leftarrow\text{Sensitivity Filter}$; \\
			\hspace*{1em}~$x\leftarrow\text{Optimality Criteria Update}$\\
			\hspace*{1em}~$\texttt{rchange}\leftarrow
			\dfrac{\max(|x-x_{\text{old}}|)}{\max(x_{\text{old}})}$;  \\
			\hspace*{1em}~$\text{Plot the Updated Design}$\\
			~\textbf{end while}
		\end{tcolorbox}
	\end{minipage}
	\hfill
	\begin{minipage}[htbp]{0.48\textwidth} 
		\centering
		\resizebox{\textwidth}{!}{ 
			\begin{tikzpicture}[node distance=1.5cm, scale=0.35, every node/.style={scale=1}]
				\node (start) [startstop] {\textbf{Start}};
				\node (domain) [process, below of=start, align=center] {%
					\textbf{Domain Initialization}\\
					(Nodes, Elements, Material Properties)};
				\node (BC) [process, below of=domain, align=center]{\textbf{FE Analysis}\\ (Apply BCs, Calc $K_e$, Solve $KU = F$)};
				\node (objective) [process, below of= BC] {\textbf{Compute Objective} $C(x_e)$};
				\node (sensitivity) [process, below of=objective] {\textbf{Compute Sensitivities} $\dfrac{\partial C}{\partial x_e}$};
				\node (filtering) [process, below of=sensitivity] {\textbf{Apply Filter}};
				\node (update) [process, below of=filtering] {\textbf{Update Design}};
				\node (plot) [process, right of=update, xshift = 5cm] {\textbf{Plot Updated Design}};
				\node (loop) [decision, below of=update, yshift=-1cm] {\textbf{Converged?}};
				\node (end) [startstop, below of=loop, yshift=-1.2cm] {\textbf{End}};
				\draw [arrow] (start) -- (domain);
				\draw [arrow] (domain) -- (BC);
				\draw [arrow] (BC) -- (objective);
				\draw [arrow] (objective) -- (sensitivity);
				\draw [arrow] (sensitivity) -- (filtering);
				\draw [arrow] (filtering) -- (update);
				\draw [arrow] (update) -- (plot);
				\draw [arrow] (update) -- (loop);
				\draw [arrow] (loop.south) -- ++(0,-0.1) node[midway, right] {Yes} -- (end.north);
				\draw [arrow] (loop.west) -- ++(-11,0) coordinate(tmp1)
				-- ++(0,24.3) coordinate(tmp2)
				-- (BC.west) node[midway, above left] {No};
			\end{tikzpicture}
		}
	\end{minipage}
	\caption{Left: Pseudocode; Right: Flowchart of the SIMP method}
	\label{fig:pseudocode_flowchart}
\end{figure}
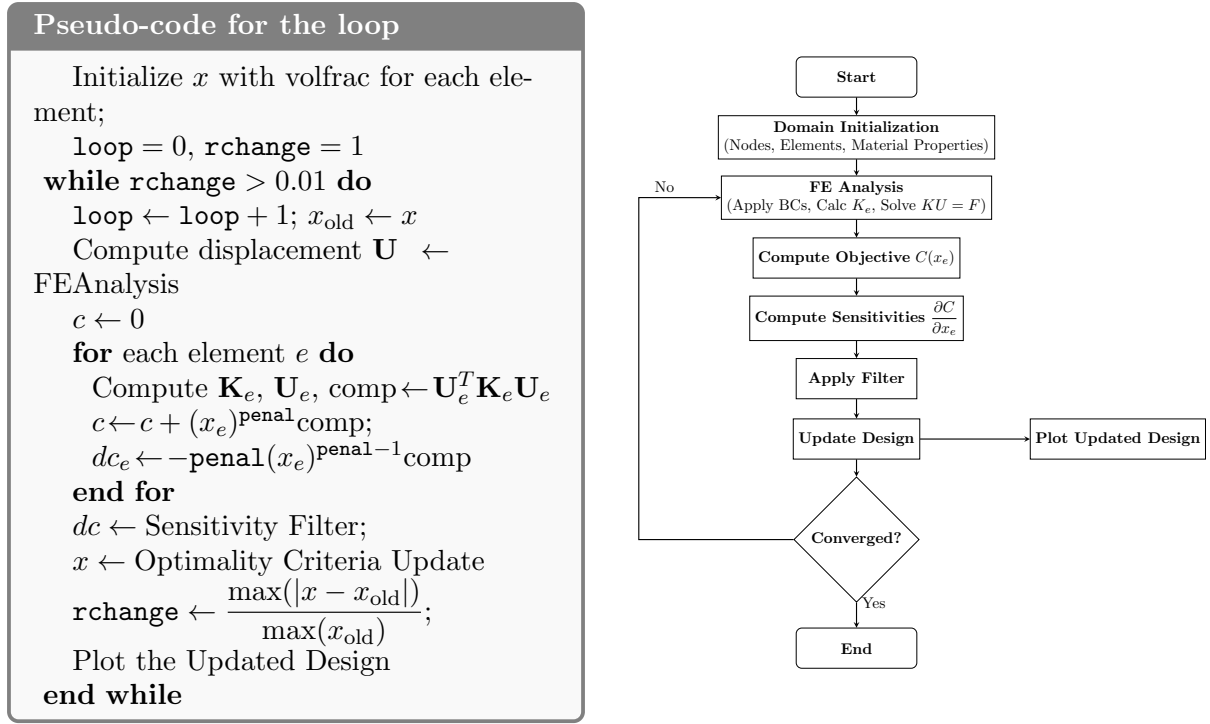

This formulation represents an $\ell_1$ norm in the transformed space and differs fundamentally from the quadratic and $\ell_2$-based measures. In contrast to energy-based aggregation, the $\ell_1$ norm promotes sparsity in the transformed displacement components, effectively emphasizing dominant deformation modes while suppressing smaller contributions.

\section{Results and Discussion}

To assess the effectiveness of the proposed formulations, a series of standard topology optimization benchmark problems were considered. For each case, the initial design domain, boundary conditions, and loading configurations were defined according to conventional benchmark settings. The performance of the three compliance formulations classical quadratic, $\ell_2$-norm, and $\ell_1$-norm was evaluated on these problems, including the bridge structure (Fig.~\ref{fig:bridge_comparison}), the cantilever beam with single loading (Fig.~\ref{fig:cant_comparison}) and dual loading conditions (Fig.~\ref{fig:twoforce_comparison}), and the well-known MBB beam (Figs.~\ref{fig:HalfMBB_comparison} and \ref{fig:FullMBB_comparison}). All simulations were conducted under identical material properties, loading conditions, and volume constraints to ensure a fair comparison among the formulations.

\begin{figure}[htbp]
	\centering
	
	\begin{minipage}{0.48\textwidth}
		\centering
		\begin{tikzpicture}[baseline={(current bounding box.north)}, scale=1.0]
			\fill[gray!50] (0,0) rectangle (4,2);
			\draw[thick] (0,0) rectangle (4,2);
			\filldraw[black] (0,0) circle (0.05);
			\filldraw[black] (4,0) circle (0.05);
			\filldraw[red] (2,0) circle (0.07);
			\draw[very thick, ->, red] (2,-0.2) -- (2,-0.6);
			\node[below] at (2,-0.6) {\footnotesize Force};
			\filldraw[black] (0,0) -- (-0.2,-0.15) -- (0.2,-0.15) -- cycle;
			\filldraw[black] (4,0) -- (3.8,-0.15) -- (4.2,-0.15) -- cycle;
			\node[below] at (0,-0.3) {\footnotesize Fixed};
			\draw[thick] (-0.35,-0.15) -- (0.35,-0.15);
			\node[below] at (4,-0.3) {\footnotesize Fixed};
			\draw[thick] (3.65,-0.15) -- (4.35,-0.15);
		\end{tikzpicture}
		
		{\textbf{Bridge Structure}}
	\end{minipage}
	\hfill
	\begin{minipage}{0.48\textwidth}
		\centering
		\begin{tikzpicture}[baseline={(current bounding box.north)}, scale=1.0]
			\fill[gray!50] (0,0) rectangle (4,2);
			\draw[thick] (0,0) rectangle (4,2);
			\foreach \y in {0.1,0.3,...,1.9}
			\draw[black] (-0.2,\y) -- (0,\y+0.1);
			\filldraw[red] (4,0) circle (0.05);
			\draw[very thick, ->, red] (4.2,0) -- (4.2,-0.4);
			\node[below] at (4.3,-0.4) {\footnotesize Force};
			\node[rotate=90] at (-0.5,1) {\footnotesize Fixed Nodes};
		\end{tikzpicture}
		
		{\textbf{Cantilever beam with one force}}
	\end{minipage}
	
	\caption{Initial bridge and cantilever beam setups. $E=1$, $\nu=0.3$, $volfrac=0.5$, $penal=3$, $\vec{F}=(0,-1)$.}
	\label{fig:bridge_cantilever}
\end{figure}
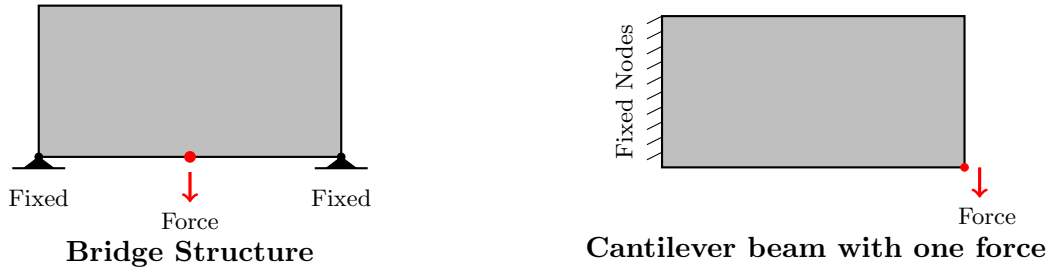
\noindent
The force node and the fixed nodes are in the following order for the bridge structure.
\begin{lstlisting}
	F(2*(nelx/2+1)*(nely+1),1) = -1;	
	Fixeddofs = union([2*(nely+1)-1:2*(nely+1),...
	                     2*(nelx+1)*(nely+1)-1:2*(nelx+1)*(nely+1)]);
\end{lstlisting}

\begin{figure}[htbp]
	\centering
	\setlength{\tabcolsep}{2pt}   
	\renewcommand{\arraystretch}{1}
	
	\begin{tabular}{|c|c|c|}
		\hline
		\textbf{Quadratic Compliance} & \textbf{$\ell_2$-Norm Compliance} & \textbf{$\ell_1$-Norm Compliance} \\
		\hline
		\includegraphics[width=0.3\textwidth,height=4.5cm,keepaspectratio]{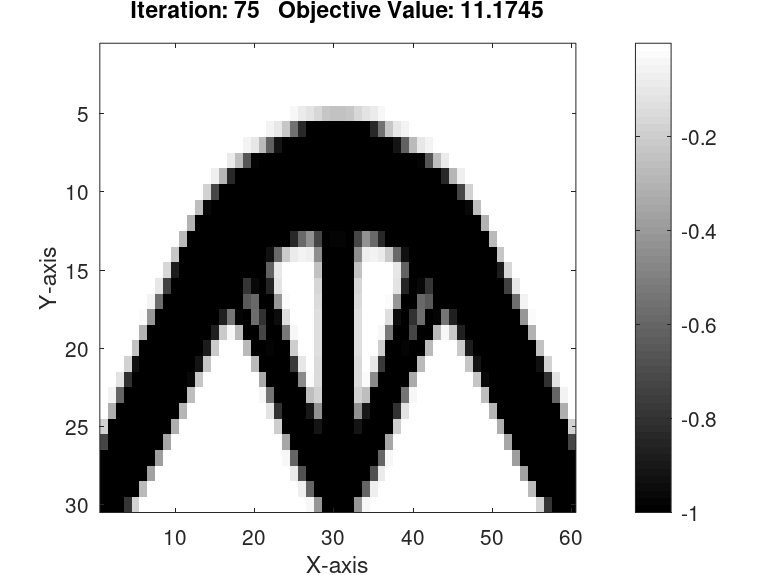} &
		\includegraphics[width=0.3\textwidth,height=4.5cm,keepaspectratio]{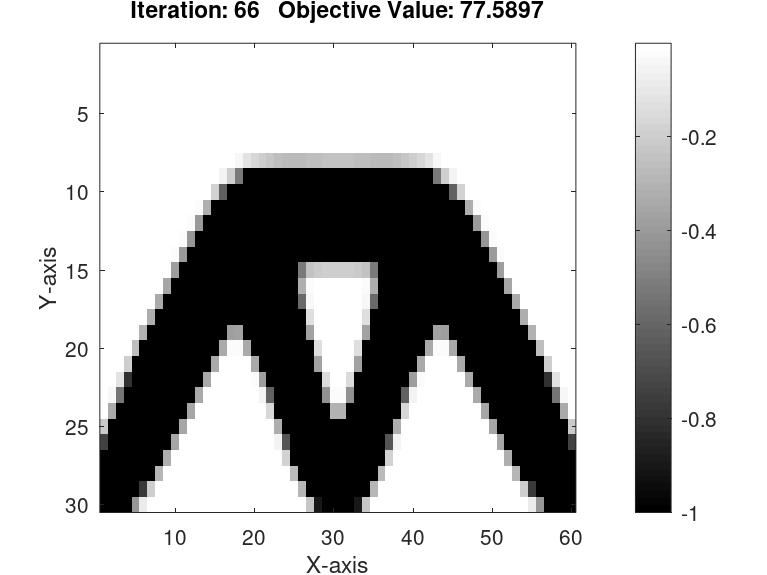} & 
		\includegraphics[width=0.3\textwidth,height=4.5cm,keepaspectratio]{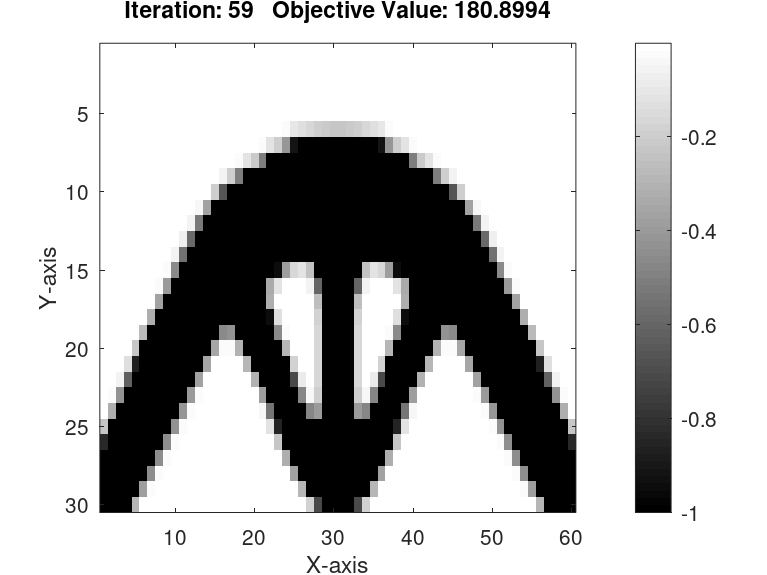} \\
		\hline
	\end{tabular}
	
	\caption{Comparison of optimized topologies for the Bridge structure by different compliance.}
	\label{fig:bridge_comparison}
\end{figure}
\noindent
The force node and the fixed nodes are in the following order for the cantilever beam with one force. Fig.~\ref{fig:bridge_cantilever} (right) and Fig.~\ref{fig:cant_comparison} represents the initial and optimized domain respectively.
\begin{lstlisting}
	F(2*(nelx+1)*(nely+1),1) = -1;
	Fixeddofs = [1:2*(nely+1)];	
\end{lstlisting}
\begin{figure}[htbp]
	\centering
	\setlength{\tabcolsep}{4pt}
	\renewcommand{\arraystretch}{1}
	\begin{tabular}{|c|c|c|}
		\hline
		\textbf{Quadratic Compliance} & \textbf{$\ell_2$-Norm Compliance} & \textbf{$\ell_1$-Norm Compliance} \\
		\hline
		\includegraphics[width=0.3\textwidth,height=4.3cm,keepaspectratio]{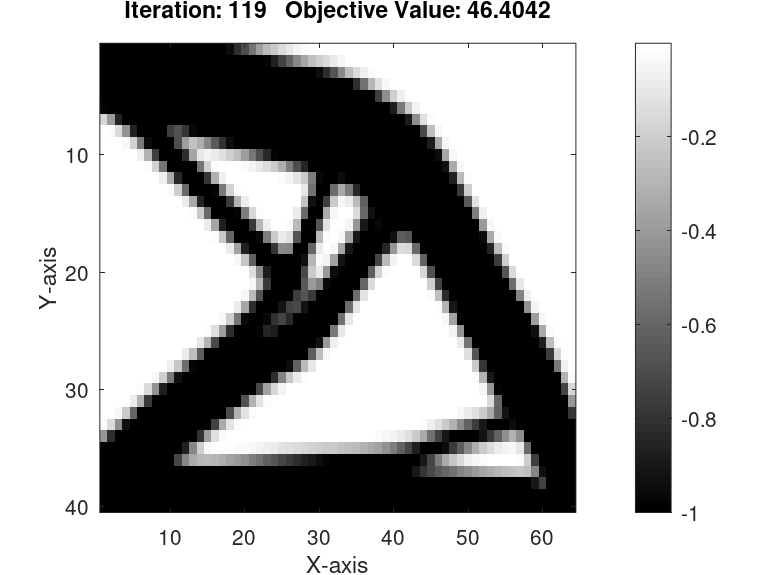} & 
		\includegraphics[width=0.3\textwidth,height=4.3cm,keepaspectratio]{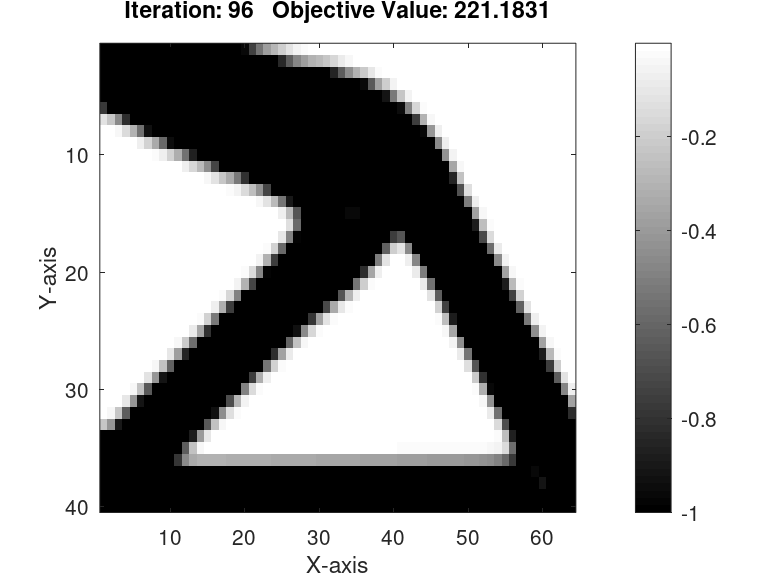} & 
		\includegraphics[width=0.3\textwidth,height=4.3cm,keepaspectratio]{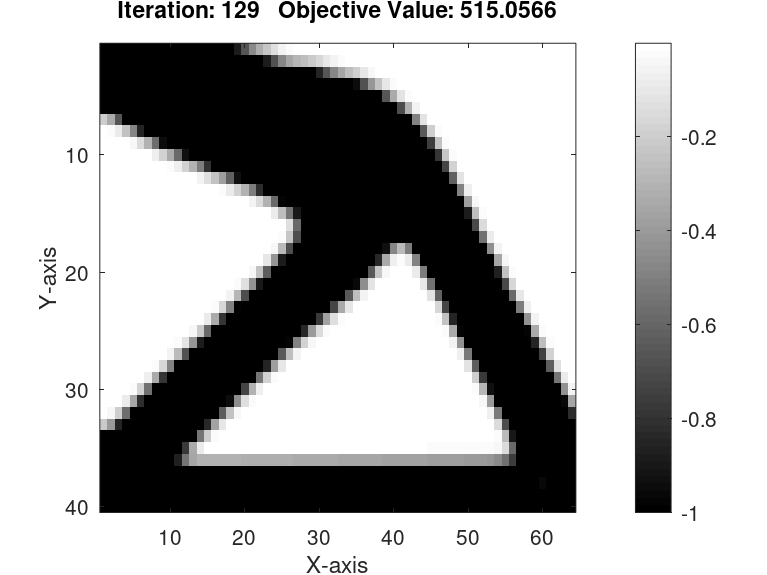} \\
		\hline
	\end{tabular}

\caption{Comparison of optimized topologies for the Cantilever beam by different compliance.}

\label{fig:cant_comparison}
\end{figure}

\noindent
\begin{minipage}{0.52\textwidth}
The force node and the fixed nodes are in the following order for the cantilever beam with two forces (upward force at the top east corner, downward force at the bottom east corner), see Fig.~\ref{fig:can_comparison}.
	
{\raggedright
	\begin{lstlisting}
		F(2*(nelx+1)*(nely+1),1) = -1;
		F(2*(nelx)*(nely+1)+2, 2) = 1;
		Fixeddofs = [1:2*(nely+1)];
	\end{lstlisting}
}
	
\end{minipage}
\hfill
\begin{minipage}{0.44\textwidth}
	\centering
	\begin{tikzpicture}[scale=0.7]
		\fill[gray!50] (0,0) rectangle (4,4);
		\draw[thick] (0,0) rectangle (4,4);
		
		\foreach \y in {0.1,0.3,...,3.9}
		\draw[black] (-0.2,\y) -- (0,\y+0.1);
		
		\filldraw[red] (4,0) circle (0.05);
		\draw[very thick, ->, red] (4.2,0) -- (4.2,-0.4);
		\node[below] at (4.3,-0.4) {\footnotesize Force};
		
		\node[rotate=90] at (-0.5,2) {\footnotesize Fixed Nodes};
		
		\filldraw[red] (4,4) circle (0.05);
		\draw[very thick, ->, red] (4.2,4.0) -- (4.2,4.4);
		\node[above] at (4.8,4.0) {\footnotesize Force};
	\end{tikzpicture}
	
	
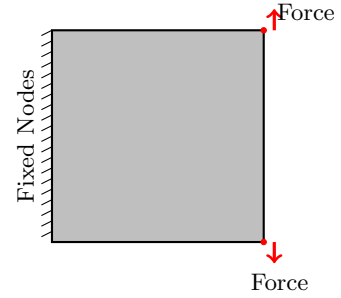
\captionof{figure}{\textbf{Cantilever Beam with two forces}}
	\label{fig:can_comparison}
\end{minipage}

\begin{figure}[htbp]
	\centering
	\setlength{\tabcolsep}{4pt}
	\renewcommand{\arraystretch}{1}
	\begin{tabular}{|c|c|c|}
		\hline
		\textbf{Quadratic Compliance} & \textbf{$\ell_2$-Norm Compliance} & \textbf{$\ell_1$-Norm Compliance} \\
		\hline
		\includegraphics[width=0.3\textwidth,height=4.3cm,keepaspectratio]{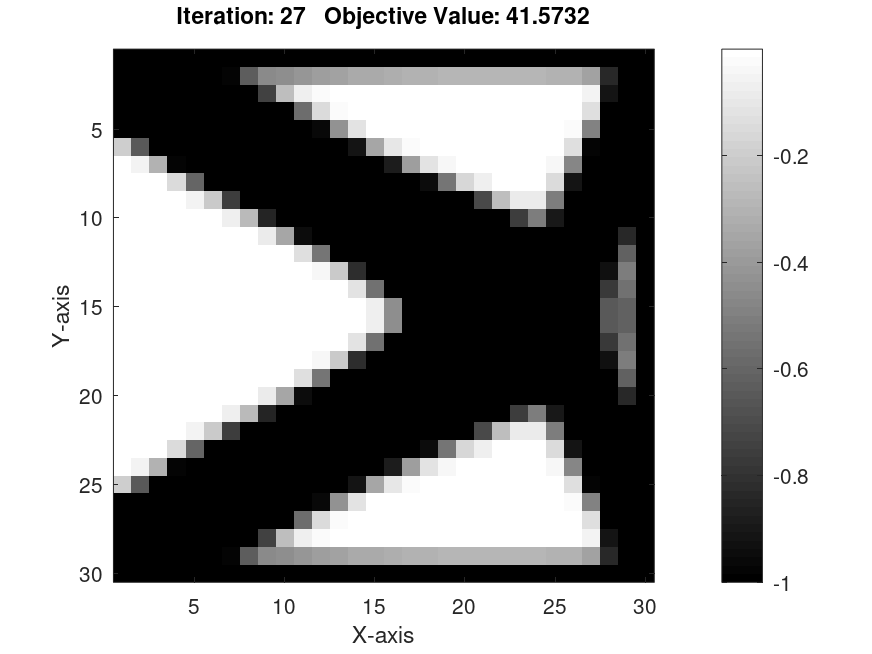} & 
		\includegraphics[width=0.3\textwidth,height=4.3cm,keepaspectratio]{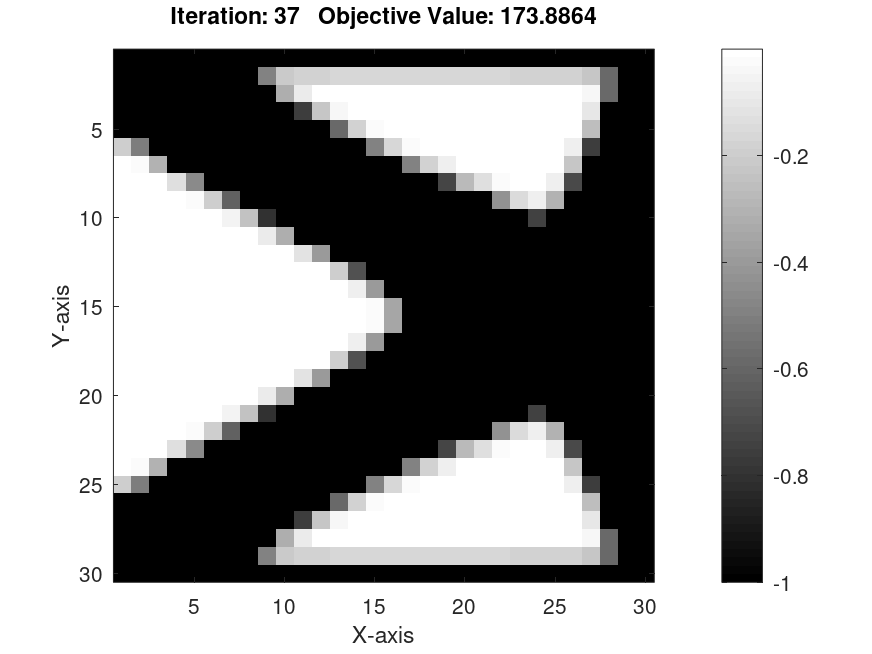} & 
		\includegraphics[width=0.3\textwidth,height=4.3cm,keepaspectratio]{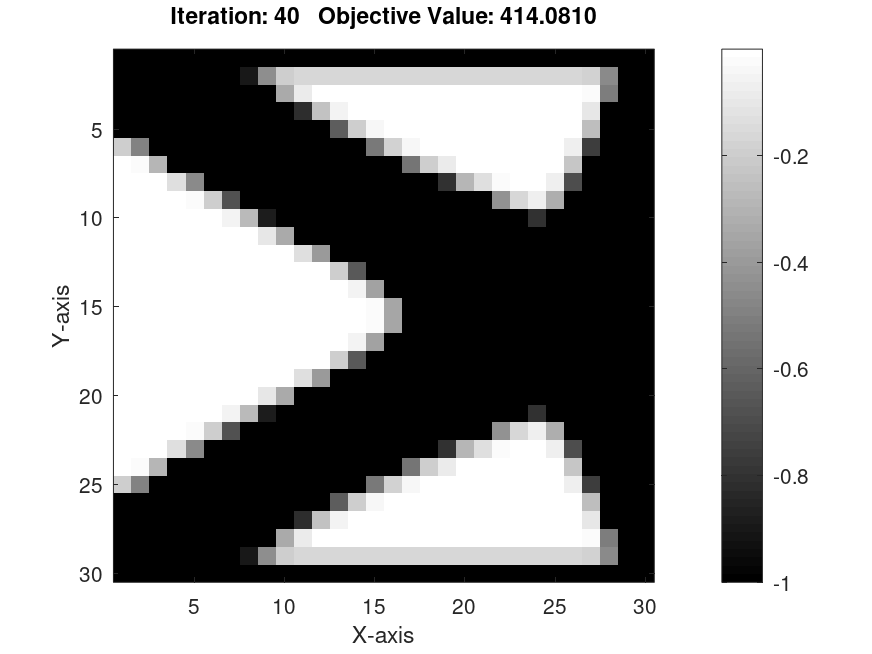} \\
		\hline
	\end{tabular}
	
	\caption{Comparison of optimized topologies for the Cantilever beam with two forces}
	
	\label{fig:twoforce_comparison}
\end{figure}

%
%
%
%
%
%
%
%
%
%
%
%
\noindent
The force node and the fixed nodes are in the following order for the half MBB-beam.
\begin{lstlisting}
	F(2,1) = -1;
	Fixeddofs = union([1:2:2*(nely+1)], [2*(nelx+1)*(nely+1)]);
\end{lstlisting}

\begin{figure}[htbp]
	\centering
	
	\begin{minipage}{0.48\textwidth}
		\centering
\begin{tikzpicture}[scale=1.5]
	
	\fill[gray!30] (0,0) rectangle (3,2);
	\draw[thick] (0,0) rectangle (3,2);
	
	\draw[very thick] (0,0) -- (0,2);
	
	\foreach \y in {0.3,0.7,1.1,1.5,1.9} {
		\draw[thick] (-0.15,\y) -- (0,\y);
	}
	
	\node[rotate=90] at (-0.5,1) {\footnotesize ($u_x=0$)};
	
	\filldraw[black] (3,0) -- (2.75,-0.2) -- (3.25,-0.2) -- cycle;
	\draw[thick] (2.85,-0.25) circle (0.05);
	\draw[thick] (3.15,-0.25) circle (0.05);
	\draw[thick] (2.6,-0.35) -- (3.4,-0.35);
	
	\draw[->, red, very thick] (0,2.3) -- (0,2);
	\node[red] at (0,2.6) {\footnotesize Load};
	
	\node[below] at (3,-0.5) {\footnotesize Rolled};	
	
\end{tikzpicture}
		\vspace{-0.5cm}
		\begin{center}
			(a)
		\end{center}
		\vspace{0.3cm}
		
		\includegraphics[width=\textwidth,height=3.5cm]{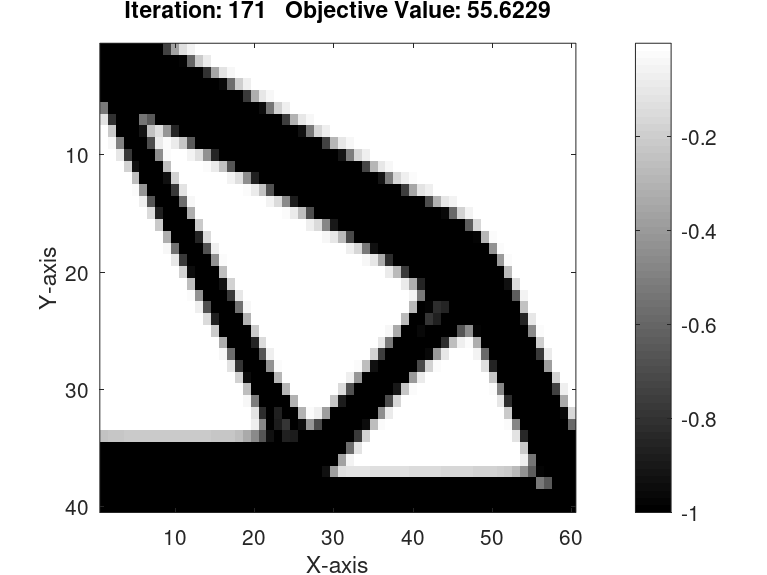}
		\begin{center}
			(b)
		\end{center}
	\end{minipage}
	\hfill
	\begin{minipage}{0.48\textwidth}
		\vspace{0.8cm}
		\centering
		
		\includegraphics[width=\textwidth,height=3.5cm]{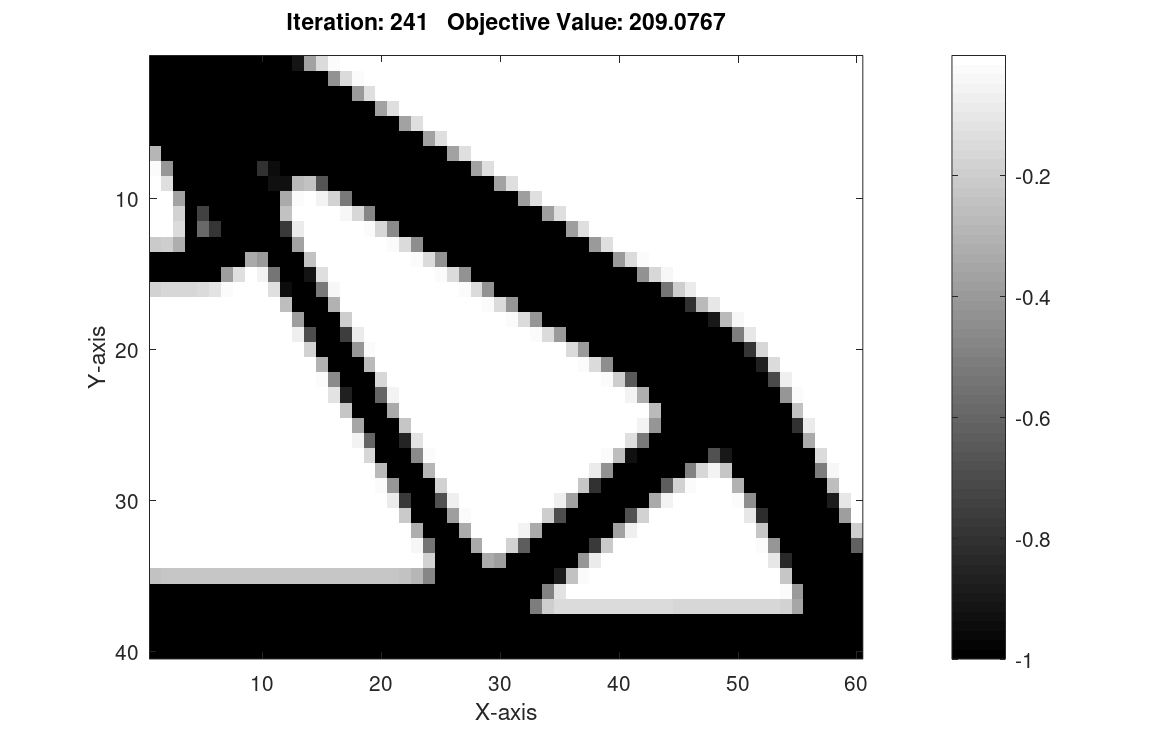}
		\begin{center}
			(c)
		\end{center}
		\vspace{1cm}
		
		\includegraphics[width=\textwidth,height=3.5cm]{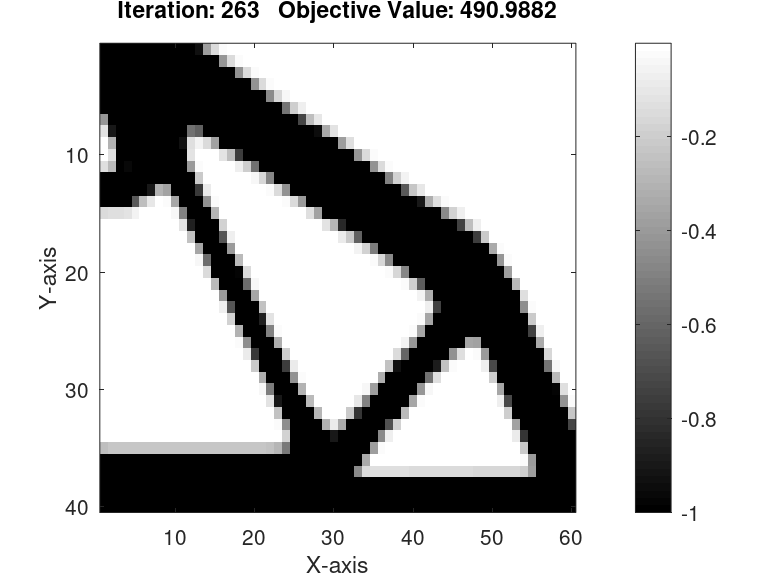}
		\begin{center}
			(d)
		\end{center}
		
	\end{minipage}
	
	\caption{Comparison of optimized topologies for the half MBB-beam by different compliance. (a) initial half MBB-beam with boundary conditions, (b) optimized full MBB-beam by classical SIMP, (c) $\ell_2$-norm compliance, (d) $\ell_1$-norm compliance.}
	\label{fig:HalfMBB_comparison}
	
\end{figure}
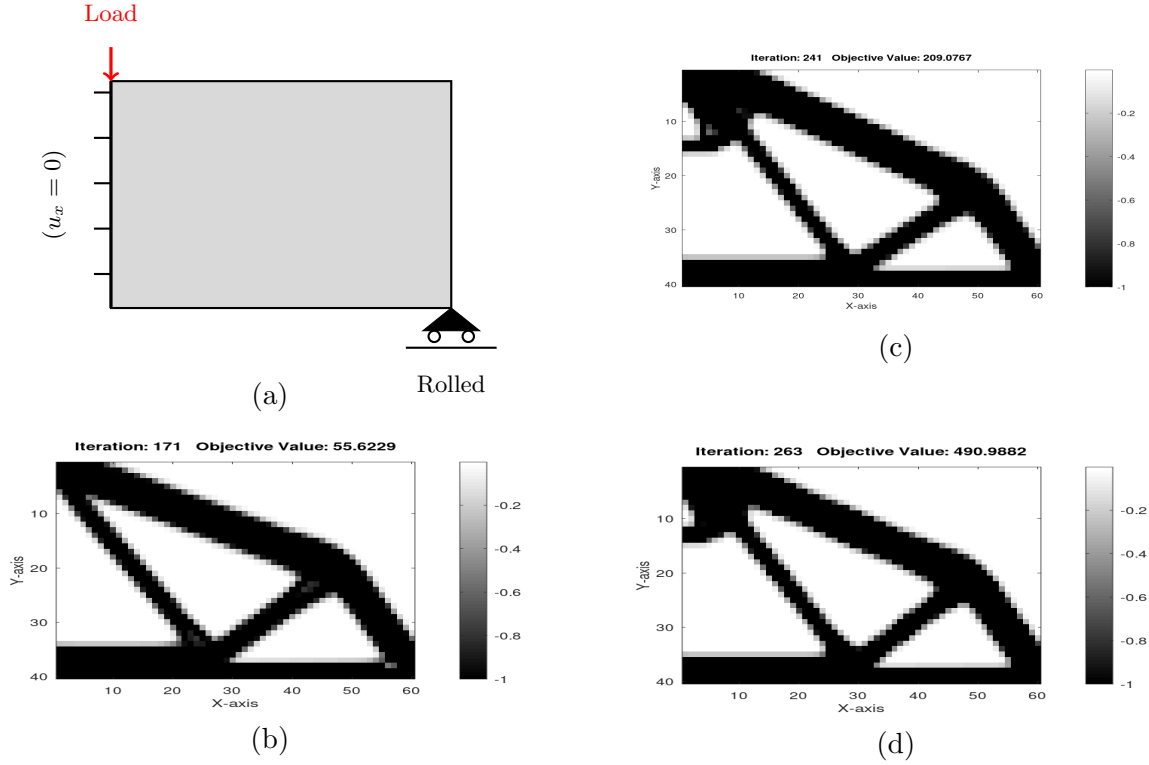

The force node and the fixed nodes are in the following order for the full MBB-beam and the optimized shapes are given in the Fig.~\ref{fig:FullMBB_comparison}.
\begin{lstlisting}
	top_middle = (nelx/2)*(nely+1)+1 ;
	F(2*top_middle) = -1;
	Fixeddofs = union([2*(nely+1)], [2*(nelx+1)*(nely+1)-1:2*(nelx+1)*(nely+1)]);
\end{lstlisting}

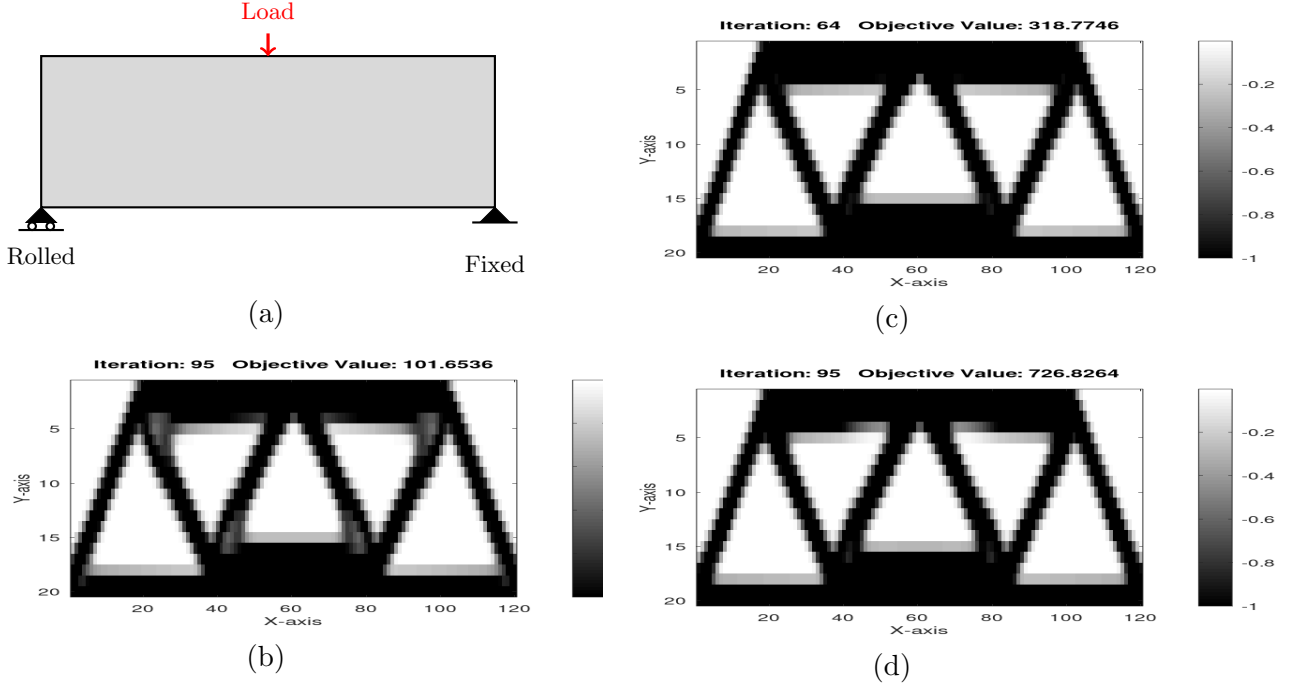
\begin{figure}[htbp]
	\centering
	
\begin{minipage}{0.48\textwidth}
	\centering
	\begin{tikzpicture}[scale=1.0]
		
		\fill[gray!30] (0,0) rectangle (6,2);
		\draw[thick] (0,0) rectangle (6,2);
		
		\filldraw[black] (0,0) -- (-0.2,-0.2) -- (0.2,-0.2) -- cycle;
		\draw[thick] (-0.3,-0.3) -- (0.3,-0.3);
		\draw[thick] (-0.12,-0.25) circle (0.05);
		\draw[thick] (0.12,-0.25) circle (0.05);	
		\node[below] at (0,-0.4) {\footnotesize Rolled};
		
		\filldraw[black] (6,0) -- (5.8,-0.2) -- (6.2,-0.2) -- cycle;
		\draw[thick] (5.7,-0.2) -- (6.3,-0.2);
		\node[below] at (6,-0.5) {\footnotesize Fixed};
		
		\draw[->, red, very thick] (3,2.3) -- (3,2);
		\node[red] at (3,2.6) {\footnotesize Load};
		\vspace{-0.5cm}	
	\end{tikzpicture}
       \begin{center}
	    	(a)
    	\end{center}
	\vspace{0.3cm}
	
	\includegraphics[width=9.5cm,height=3.5cm]{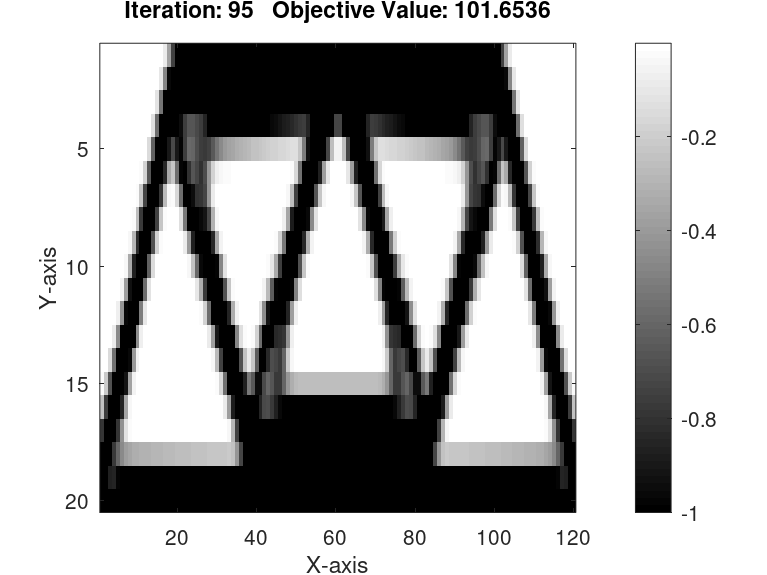}
	\begin{center}
		(b)
	\end{center}
\end{minipage}
	\hfill
	\begin{minipage}{0.48\textwidth}
		\vspace{0.5cm}
		\centering
		
		\includegraphics[width=9.5cm,height=3.5cm]{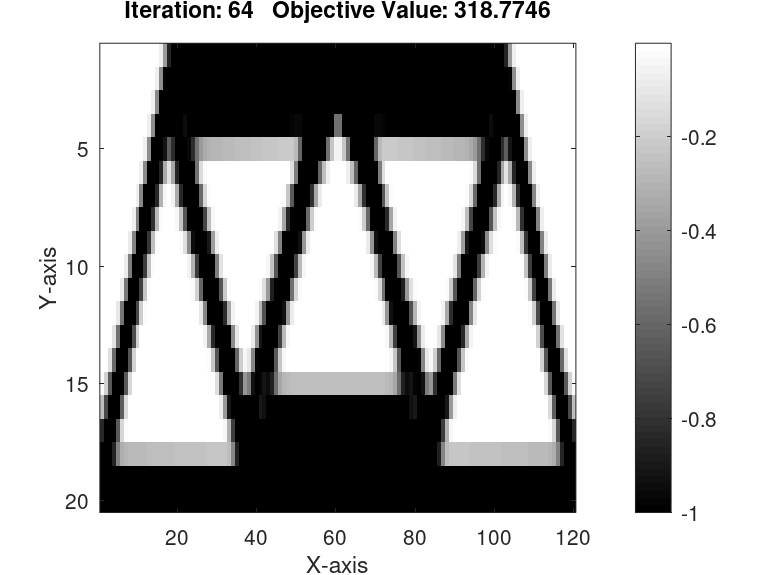}
		\begin{center}
			(c)
		\end{center}
		\vspace{0.35cm}
		
		\includegraphics[width=9.5cm,height=3.5cm]{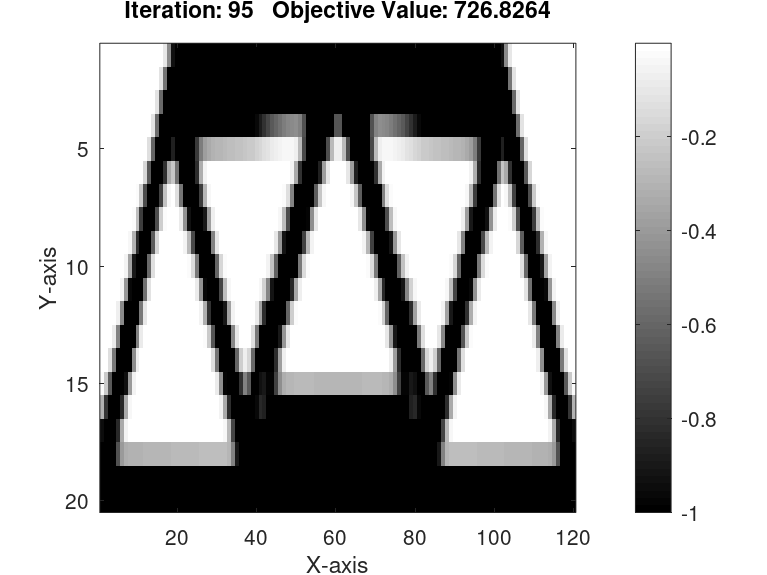}
		\begin{center}
			(d)
		\end{center}
		
	\end{minipage}
	
	\caption{Comparison of optimized topologies for the Full MBB-beam by different compliance. (a) initial full MBB-beam with boundary conditions, (b) optimized full MBB-beam by classical SIMP, (c) $\ell_2$-norm compliance, (d) $\ell_1$-norm compliance.}
	\label{fig:FullMBB_comparison}
	
\end{figure}

The following table presents a comparison of these compliance measures in terms of their mathematical form, sensitivity behavior, and impact on optimization outcomes.

\begin{table}[htbp]
	\centering
	\caption{Comparison of compliance formulations}
	\label{tab:compliance_comparison}
	\begin{tabular}{llll}
		\toprule
		\textbf{Feature} & \textbf{Quadratic Compliance} & $\ell_2$-Norm & $\ell_1$-Norm \\
		\midrule
		Mathematical form & $U^\top K U$ & $\sqrt{U^\top K U}$ & $\|K^{1/2}U\|_1$ \\
		Nature & Smooth, quadratic & Scaled quadratic & Non-smooth \\
		Sensitivity behavior & Well-distributed & Moderately scaled & Sparse/localized \\
		Topology characteristics & Distributed members & Intermediate behavior & Sparse structures \\
		Optimization stability & High & Moderate & Lower \\

		\bottomrule
	\end{tabular}
\end{table}

The classical quadratic compliance formulation produces smooth and well-connected structural layouts with distributed load paths. This behavior is consistent with its interpretation as total strain energy, where contributions from all deformation modes are accumulated in a quadratic manner. As a result, the optimization process favors designs that evenly distribute stress throughout the domain, leading to stable and physically intuitive topologies.

The $\ell_2$-norm compliance, defined as the square root of the quadratic form, yields topologies that are qualitatively similar to those obtained from the classical formulation. Since this transformation is monotonic, the overall structural patterns remain largely unchanged. However, subtle differences can be observed in the thickness and distribution of members. These variations arise due to the rescaling of sensitivities, which slightly alters the optimization trajectory while preserving the fundamental characteristics of the solution. In contrast, the $\ell_1$-norm compliance formulation produces markedly different topologies characterized by sparse and highly localized structural members. The resulting designs exhibit fewer but more pronounced load-carrying paths, often resembling truss like structures. This behavior can be attributed to the inherent sparsity-promoting property of the $\ell_1$ norm, which penalizes distributed contributions and instead favors concentration in a limited number of dominant deformation modes. Consequently, the material distribution becomes highly non-uniform, with clear emphasis on critical load paths. From an optimization perspective, the quadratic and $\ell_2$ formulations exhibit smooth objective landscapes, which contribute to stable convergence and well-behaved sensitivity fields. In contrast, the $\ell_1$ formulation introduces non-smoothness into the optimization problem, which can lead to increased sensitivity to numerical parameters and potentially slower or less stable convergence. Nevertheless, this formulation provides a useful mechanism for generating structurally efficient designs with minimal material usage along principal load paths.

Overall, the results demonstrate that although the three compliance measures originate from the same physical model, their mathematical structure significantly influences the resulting topologies. The quadratic formulation favors uniform stiffness distribution, the $\ell_2$ formulation maintains similar characteristics with slight variations, and the $\ell_1$ formulation promotes sparsity and localization. These findings highlight the critical role of objective function selection in topology optimization and suggest that alternative norm-based formulations can be effectively used to tailor structural designs to specific performance requirements.

\section{Conclusions}

This work examined the impact of alternative compliance formulations on topology optimization. The classical quadratic and $\ell_2$-norm formulations produced similar, well-distributed structural layouts, while the $\ell_1$-norm formulation led to sparse and highly localized designs. These results highlight that the choice of objective function significantly influences the resulting topology.

\section*{Conflicts of Interest}

On behalf of all authors, the corresponding author declares that there are no conflicts of interest related to the publication of this work.

\section*{Data Availability Statement}
The data presented in this study are available on request from the corresponding author due to some of the data involves privacy.

\section*{Grammar and Readability Disclosure}
The authors used ChatGPT (OpenAI)~\cite{chatgpt2026} for grammar checking and language refinement during manuscript preparation. The authors take full responsibility for the content.

\section*{Acknowledgements}
The authors would like to thank SRM University, Andhra Pradesh for providing the fellowship and lab facilities.

    \bibliographystyle{unsrt}
	\bibliography{reference}

\end{document}